\documentclass[11pt]{article}
\usepackage[left=1in,top=1in,right=1in,bottom=1in]{geometry}
\usepackage{times}

\usepackage{verbatim}
\usepackage{amsmath}
\usepackage{amssymb}
\usepackage{amsthm}
\usepackage{rotating}
\usepackage{algorithm}
\usepackage[noend]{algpseudocode}

\usepackage{silence}
\usepackage{tabularx}
\usepackage{graphicx}
\usepackage{url}
\usepackage{multirow}
\usepackage{subfig}

\usepackage[noend]{algpseudocode}
\usepackage{siunitx}

\usepackage{array}      

\newtheorem{theorem}{Theorem}
\newtheorem{corollary}{Corollary}
\newtheorem{definition}{Definition}

\newtheorem{proposition}{Proposition}

\DeclareMathOperator*{\argmin}{arg\,min}

\allowdisplaybreaks

\begin{document}

\title{Consistent Opponent Modeling in Imperfect-Information Games}

\author{Sam Ganzfried\\
Ganzfried Research \\
sam.ganzfried@gmail.com
}

\date{\vspace{-5ex}}

\maketitle

\begin{abstract}
The goal of agents in multi-agent environments is to maximize total reward against the opposing agents that are encountered. Following a game-theoretic solution concept, such as Nash equilibrium, may obtain a strong performance in some settings; however, such approaches fail to capitalize on historical and observed data from repeated interactions against our opponents. Opponent modeling algorithms integrate machine learning techniques to exploit suboptimal opponents utilizing available data; however, the effectiveness of such approaches in imperfect-information games to date is quite limited. We show that existing opponent modeling approaches fail to satisfy a simple desirable property even against static opponents drawn from a known prior distribution; namely, they do not guarantee that the model approaches the opponent's true strategy even in the limit as the number of iterations approaches infinity. We develop a new algorithm that achieves this property given sufficiently informative observations and runs efficiently by solving a convex minimization problem based on the sequence-form game representation.
\end{abstract}

\section{Introduction}
\label{se:intro}
There has been a surge of progress in development of efficient algorithms for computation (and approximation) of game-theoretic solution concepts in imperfect-information games. However, such approaches do not account for the availability of data based on observations of our opponents' play. Against suboptimal opponents we can obtain significantly higher payoffs in practice by integrating techniques that utilize current and historical data. Opponent modeling contains many additional challenges not present in computation of traditional game-theoretic solution concepts, and research in this area is significantly more nascent. As we will see, existing opponent modeling approaches fail to accurately model static opponents (whose strategy does not change over time) even in small two-player games. This is perhaps the simplest possible setting, and is only the tip of the iceberg for research in this field. It is important to solve this setting if we want to have any hope to address further complexities, such as dynamic opponents and multiple opposing agents (whose strategies may or may not be correlated).

In perfect-information games, opponent modeling is relatively straightforward. First consider the case of simultaneous-move games, often represented using the normal form. Let $S_i$ denote the set of pure strategies for player $i$, $S$ denote the set of strategy profiles (vectors of strategies for each player), and $u_i: S \rightarrow \mathbb{R}$ the utility function for player $i$. For now suppose the game has two players and that we play the role of player 1. In the absence of any historical data on our opponent, we can use a prior distribution that is Dirichlet with mean on the mixed strategy that plays each pure strategy with equal probability. This can be done by initializing $m^0_{s_{2,j}} = \alpha$ for all $s_{2,j} \in S_{2},$ where $\alpha$ is a positive constant and $S_{2}$ is the pure strategy space of the opponent. At each iteration $t$, we assume that the opponent will play
pure strategy $s_{2,j}$ with probability
$\frac{m^{t-1}_{s_{2,j}}}{\sum_k m^{t-1}_{s_{2,k}}}$,
and we can simply play any best response to this strategy.
If the opponent plays $s_{2,j}$ at time $t$, then we set $m^t_{s_{2,j}} = m^{t-1}_{s_{2,j}} + 1$, and otherwise we set $m^t_{s_{2,j}} = m^{t-1}_{s_{2,j}}.$ Then our new opponent model will again be the mean of the new Dirichlet posterior distribution given our prior and observations. By the law of large numbers, the opponent model will approach the opponent's true static mixed strategy as the number of game iterations $t$ approaches infinity. If historical data is available, then we could use a Dirichlet prior with a different mean strategy by modifying the initial values $m^0_{s_{2,j}}$ accordingly. Thus, in two-player normal-form games, we can straightforwardly perform efficient opponent modeling and exploitation against static mixed strategies. This approach can also be straightforwardly extended to apply to more than two players whose strategies are played independently. In perfect-information sequential games, opponent modeling is also straightforward. Such games are typically represented as trees. We can simply apply the same procedure as above independently at each decision node of the tree. This also applies whether or not chance moves are allowed in the game tree.

Opponent modeling in imperfect-information games is significantly more challenging because we do not observe the private information of the opponent and may not be sure which nodes of the tree are reached: so we cannot simply maintain counters at each of the opponent's decision nodes. Imperfect-information games are modeled using extensive-form game trees, where play proceeds from the root node to a terminal leaf node at which point all players receive payoffs. Each non-terminal node has an associated player (possibly \emph{chance}) that makes the decision at that node. These nodes are partitioned into \emph{information sets}, where the player whose turn it is to move cannot distinguish among the states in the same information set. Therefore, in any given information set, a player must choose actions with the same distribution at each state contained in the information set. If no player forgets information that they previously knew, we say that the game has \emph{perfect recall}. A (mixed) \emph{strategy} for player $i,$ $\boldsymbol{\sigma}_i \in \Sigma_i,$ is a function that assigns a probability distribution over all actions at each information set belonging to $i$. 

To perform opponent modeling we assume that we are in the \emph{repeated game setting} where the game is repeated for some number of iterations. This number may be finite, infinite, or unknown (and potentially infinite). Defining repeated games in the normal-form setting is relatively straightforward. At each iteration $t$, we know our own mixed strategy $\boldsymbol{\sigma}^t_i,$ as well as the vector of pure strategies selected by all players $\mathbf{s}^t$ (though of course we do not know opponents' mixed strategies). We remember all of these strategies for all subsequent iterations. Therefore, at iteration $t$, our strategy can be dependent on $\boldsymbol{\sigma}^1_i,\ldots,\boldsymbol{\sigma}^{t-1}_i$ as well as $\mathbf{s}^1,\ldots,\mathbf{s}^{t-1}$. In perfect-information sequential games, we know our own mixed strategy $\boldsymbol{\sigma}^t_i,$ as well as the actions taken by all players at nodes along the path of play (we do not know what an opposing player would have done at a node that was not reached). Again, at time $t$ our strategy can be dependent on these values for times from 1 to $t-1.$

In imperfect-information games the situation is more complex. For some game iterations we may observe the full trajectory of nodes along the path of play, similarly to the perfect-information case; however, in other iterations we may only know a set of paths of play that are consistent with the true trajectory. For example, in poker if player 1 bets and player 2 folds, then player 1 wins the hand and player 2's card is not revealed; however, if player 1 bets and player 2 calls, then a \emph{showdown} occurs and both cards are revealed, with the pot going to the player with the superior hand. In general, imperfect-information games exhibit partial observability of opponents' private information, which may lead to uncertainty about the path of play taken. In order to define the concept of a repeated imperfect-information game, we must introduce an \emph{observability function}. Recall that if $L$ is the set of leaf nodes, for each player $i$ there is a function $u_i: L \rightarrow \mathbb{R},$ that gives the utility to player $i$ if $\ell \in L$ is reached. Note that we can represent a leaf node by the sequence of actions along the trajectory from the root node. Let $P(L)$ be the power set of $L.$ For each player, the observability function $o_i : L \to P(L)$ maps each leaf node $\ell \in L$ to the set of trajectories that are possible from player $i$'s perspective when $\ell$ is reached (note that player $i$ does not necessarily know $\ell,$ only $o_i(\ell)$). Player $i$'s strategy at time $t$ can now be dependent on their mixed strategies $\boldsymbol{\sigma}^j_i$ for $j = 1,\ldots,t-1$ as well as $o_i(\ell_j)$ for $j = 1,\ldots,t-1.$ 

Now that we have defined the repeated game setting for imperfect-information games, we can consider opponent modeling algorithms for this setting. It turns out that the general goal is still the same; given a prior distribution (possibly an uninformative Dirichlet distribution whose mean selects each available action with equal probability) and our observations made during gameplay, we would like our model for the opponent's strategy to be the mean of the posterior distribution. This is justified by the following theorem and corollary~\cite{Ganzfried18b:Bayesian}. Theorem~\ref{th:mean} shows that the expected payoff of a mixed strategy $\boldsymbol{\sigma}_i$ against a probability distribution with pdf $f_{-i}$ over mixed strategies equals the expected payoff of playing $\boldsymbol{\sigma}_i$ against the mean $\overline{f_{-i}}$ of $f_{-i}$. This applies to both normal-form and extensive-form games (with perfect and imperfect information) as well as to any number of players ($\boldsymbol{\sigma}_{-i}$ can be a joint strategy profile for all opposing agents). From this result Corollary~\ref{co:posterior-mean} immediately follows if we let $f_{-i}$ be the posterior probability distribution for the opponent's strategy $\boldsymbol{\sigma}_{-i}$ given observations $\mathbf{x}$, denoted as $p(\boldsymbol{\sigma}_{-i} | \mathbf{x})$.

For the following theorem and corollary, a mixed strategy is a
probability distribution over normal-form pure strategies, and
the mean is taken in these coordinates. In two-player
perfect-recall extensive-form games, one can equivalently
average the opponent's realization plans. Averaging behavioral
action probabilities separately at each information set does
not generally preserve expected payoffs. With multiple
opponents, the averaging must retain their joint distribution
over pure strategy profiles.

\begin{theorem}
$$u_i(\boldsymbol{\sigma}_i, \overline{f_{-i}}) = u_i(\boldsymbol{\sigma}_i, f_{-i}).$$ That is, the payoff against the mean of a strategy distribution equals the payoff against the full distribution~\cite{Ganzfried18b:Bayesian}. 
\label{th:mean}
\end{theorem}

\begin{corollary}
$u_i(\boldsymbol{\sigma}_i, \overline{p(\boldsymbol{\sigma}_{-i} | \mathbf{x})}) = u_i(\boldsymbol{\sigma}_i, p(\boldsymbol{\sigma}_{-i} | \mathbf{x}))$~\cite{Ganzfried18b:Bayesian}. 
\label{co:posterior-mean}
\end{corollary}

These results imply the meta-procedure for optimizing performance against an opponent in a two-player imperfect-information game that is depicted in Algorithm~\ref{al:meta}~\cite{Ganzfried18b:Bayesian}. In the algorithm, $p_0$ denotes the prior distribution, $r_t$ is a response function (typically a best response), and $\mathbf{M}_t$ denotes our model for the opponent's strategy at time $t.$ As described earlier, this procedure is straightforward using a Dirichlet prior distribution for perfect-information games. However, for imperfect-information games computation of the mean of the posterior distribution is much more challenging. The most effective approach is to approximate the posterior distribution by sampling a relatively small number of (mixed) strategies $k$ independently from the prior distribution in advance of gameplay and calculating the posterior probabilities only for these $k$ strategies~\cite{Southey05:Bayes}. They consider three different approaches for constructing the opponent model from these samples called Bayesian Best Response (BBR), Max A Posteriori Response, and Thompson's Response. Subsequent research has shown that BBR outperforms the other two approaches when compared against an exact (unsampled) posterior best response in a simplified setting~\cite{Ganzfried18b:Bayesian}. BBR constructs the opponent model by assuming that the opponent plays each of the $k$ sampled strategies in proportion to their posterior probabilities. 

\begin{algorithm}[!ht]
\caption{Meta-algorithm for Bayesian opponent exploitation in two-player imperfect-information games~\cite{Ganzfried18b:Bayesian}}
\label{al:meta} 
\textbf{Inputs}: Prior distribution $p_0$, response functions $r_t$ for $0 \leq t \leq T$
\begin{algorithmic}
\State $\mathbf{M}_0 \gets \overline{p_0(\boldsymbol{\sigma}_{-i})}$
\State $\mathbf{R}_0 \gets r_0(\mathbf{M}_0)$
\State Play according to $\mathbf{R}_0$
\For {$t = 1$ to $T$}
\State $x_t \gets $ observations of opponent's play at time step $t$
\State $p_t \gets \mbox{ posterior distribution of opponent's strategy} \mbox{ given prior } p_0 \mbox{ and observations } x_1,\ldots,x_t.$
\State $\mathbf{M}_t \gets$ mean of $p_t$
\State $\mathbf{R}_t \gets r_t(\mathbf{M}_t)$
\State Play according to $\mathbf{R}_t$
\EndFor
\end{algorithmic}
\end{algorithm}

\section{Consistent opponent modeling}
\label{se:consistent}
Sampled BBR~\cite{Southey05:Bayes} and its extension to multiplayer games~\cite{Ganzfried24:Opponent} have proven to be effective in experiments, significantly outperforming Nash equilibrium strategies against suboptimal opposing agents. However, we will show that this approach fails to satisfy a natural property even when playing against a static opponent drawn from the prior distribution. 

\begin{definition}
\label{de:consistent}
Suppose $\mathbf{M}_1, \mathbf{M}_2, \ldots$ are the opponent models generated by an algorithm when playing against an opponent using a static strategy $\boldsymbol{\sigma}^*_{-i}$. The algorithm is consistent if, for every possible set of sampled strategies used internally by the algorithm,
$\lim_{t\to\infty} \mathbf{M}_t = \boldsymbol{\sigma}^*_{-i}$.
\end{definition}

\begin{proposition}
\label{pr:consistent}
BBR is not consistent.
\end{proposition}
\begin{proof}
If we can show that BBR is not consistent in a normal-form game, this directly implies that it is also not consistent in extensive-form games of perfect and imperfect information.
Consider the game of Rock-Paper-Scissors and suppose that $\boldsymbol{\sigma}^*_{-i} = (0.8,0.1,0.1)$ for (R,P,S). Suppose that we are using BBR with $k = 3$ sampled strategies: (0.5,0.3,0.2), (0.3,0.5,0.2), (0.2,0.3,0.5). BBR chooses the opponent model to play each of the sampled strategies in proportion to their posterior value; therefore, the opponent model is always a convex combination of the sampled strategies.
Since $\boldsymbol{\sigma}^*_{-i}$ is clearly not close to any convex combination of the sampled strategies, clearly $\mathbf{M}_t$ cannot approach $\boldsymbol{\sigma}^*_{-i}$. 
\end{proof}

As the proof of Proposition~\ref{pr:consistent} shows, clearly BBR is constrained to only selecting opponent models that are convex combinations of the sampled strategies. So if the true strategy $\boldsymbol{\sigma}^*_{-i}$ is outside the convex hull of the sampled strategies, BBR has no hope of the opponent model approaching the true strategy even if the game is repeated infinitely many times. Of course, in a normal-form game such as Rock-Paper-Scissors we could include all pure strategies of the game in our set of ``samples,'' which would ensure that every mixed strategy is within the convex hull of the samples. However, in extensive-form games there are exponentially many pure strategies in the size of the game tree and it is not feasible to include all of them. In fact, Proposition~\ref{pr:consistent2} shows that BBR may not be consistent even if $\boldsymbol{\sigma}^*_{-i}$ is in the convex hull of the sampled strategies.

\begin{proposition}
\label{pr:consistent2}
BBR is not consistent even if $\boldsymbol{\sigma}^*_{-i}$ is in the convex hull of the sampled strategies.
\end{proposition}
\begin{proof}
Again consider Rock-Paper-Scissors. Suppose that $\boldsymbol{\sigma}^*_{-i} = \left(\frac{1}{3},\frac{1}{3},\frac{1}{3}\right),$ and that there are 3 sampled strategies: $\mathbf{s}_1 = (0.2,0.4,0.4),$ $\mathbf{s}_2 = (0.6,0.3,0.1),$ $\mathbf{s}_3 = (0.2,0.3,0.5).$ Then $\boldsymbol{\sigma}^*_{-i} = \frac{1}{3} \mathbf{s}_1 + \frac{1}{3} \mathbf{s}_2 + \frac{1}{3} \mathbf{s}_3,$ and is in the convex hull of the samples. As $t \to \infty$, the empirical frequencies of observations converge to those induced by $\boldsymbol{\sigma}^*_{-i}$ by the law of large numbers. Consequently, asymptotically the posterior weight assigned to a sampled strategy $(x_1,x_2,x_3)$ is proportional to $x_1^{t/3}x_2^{t/3}x_3^{t/3}$. Denote this distribution by $q.$

$$q(x_1,x_2,x_3) \propto x_1^{\frac{t}{3}}x_2^{\frac{t}{3}}x_3^{\frac{t}{3}}$$
$$q(0.2,0.4,0.4) \propto 0.2^{\frac{t}{3}}0.4^{\frac{t}{3}}0.4^{\frac{t}{3}} = 0.032^{\frac{t}{3}}$$
$$q(0.6,0.3,0.1) \propto 0.6^{\frac{t}{3}}0.3^{\frac{t}{3}}0.1^{\frac{t}{3}} = 0.018^{\frac{t}{3}}$$
$$q(0.2,0.3,0.5) \propto 0.2^{\frac{t}{3}}0.3^{\frac{t}{3}}0.5^{\frac{t}{3}} = 0.03^{\frac{t}{3}}$$
So BBR will select $s_1$ with probability 
$$\frac{0.032^{\frac{t}{3}}}{0.032^{\frac{t}{3}} + 0.018^{\frac{t}{3}} + 0.03^{\frac{t}{3}}}
= \frac{1}{1 + \left(\frac{0.018}{0.032}\right)^{\frac{t}{3}} + \left(\frac{0.03}{0.032}\right)^{\frac{t}{3}}}$$

$$\lim_{t \rightarrow \infty} \frac{1}{1 + \left(\frac{0.018}{0.032}\right)^{\frac{t}{3}} + \left(\frac{0.03}{0.032}\right)^{\frac{t}{3}}}
= \frac{1}{1 + 0 + 0} = 1.$$

So as $t$ goes to infinity, BBR will select $\mathbf{s}_1$ with probability 1 and $\mathbf{s}_2,\mathbf{s}_3$ with probability 0. Since $s_1$ differs from $\boldsymbol{\sigma}^*_{-i},$ the algorithm is not consistent.
\end{proof}

To summarize, we have shown that BBR is not consistent even if the opponent is playing a static strategy drawn from the known prior distribution that is in the convex hull of the sampled strategies.
This is true even if the game has two players and is small. As the proof of Proposition~\ref{pr:consistent2} shows, unless the limiting posterior values are precisely equal for multiple sampled strategies, BBR will ultimately choose the opponent model to be one of the sampled strategies with probability 1. When the number of samples is small compared to the size of the game, which is the case in practice, we would need to get very lucky with our sampling for one of the samples to be very close to $\boldsymbol{\sigma}^*_{-i}.$ As the number of samples goes to infinity we would be able to guarantee that a sampled strategy is arbitrarily close to $\boldsymbol{\sigma}^*_{-i}$, assuming that the prior distribution has full support (which is the case for the Dirichlet distribution assuming all initial parameters are strictly positive). However, BBR is only practical when the number of samples is relatively small.

\section{Algorithm}
\label{se:algorithm}
We have shown that the BBR approach, despite its observed experimental success, has significant theoretical limitations. In order to present our new approach we review the \emph{sequence-form representation} for two-player extensive-form games of imperfect information~\cite{Koller94:Fast}. Rather than operate on the full normal-form pure strategy space, which has size exponential in the size of the game tree, the sequence-form works with sequences of actions along trajectories from the root node to leaf nodes. For player 1, the matrix $\mathbf{E}$ is defined where each row corresponds to an information set (including an initial row for the ``empty'' information set), and each column corresponds to an action sequence (including an initial column for the ``empty'' action sequence). In the first row of $\mathbf{E}$ the first element is 1 and all other elements are 0; subsequent rows have -1 for the entries corresponding to the action sequence leading to the root of the information set, and 1 for all actions that can be taken at the information set (and 0 otherwise). Thus $\mathbf{E}$ has dimension $(c_1 + 1) \times (d_1+1)$, where $c_i$ is the number of information sets for player $i$ and $d_i$ is the number of action sequences for player $i$. Matrix $\mathbf{F}$ is defined analogously for player 2. The vector $\mathbf{e}$ is defined to be a column vector of length $c_1+1$ with 1 in the first position and 0 in other entries, and vector $\mathbf{f}$ is defined with length $c_2+1$ analogously. The matrix $\mathbf{A}$ is defined with dimension $(d_1+1) \times (d_2+1)$ where entry $A_{ij}$ gives the payoff for player 1 when player 1 plays action sequence $i$ and player 2 plays action sequence $j$ multiplied by the probabilities of chance moves along the path of play. The matrix $\mathbf{B}$ of player 2's payoffs is defined analogously. In zero-sum games $\mathbf{B} = -\mathbf{A}.$

Given these matrices we can solve one of two linear programming problems to compute a Nash equilibrium in zero-sum extensive-form games~\cite{Koller94:Fast}. In the first formulation the primal variables $\mathbf{x}$ correspond to player 1's mixed strategy while the dual variables correspond to player 2's strategy. In the second formulation, which is the dual problem of the first formulation, the primal decision variables $\mathbf{y}$ correspond to player 2's strategy while the dual variables correspond to player 1's strategy.

\[
\begin{array}{rrl} 
&\max_{\mathbf{x},\mathbf{q}}& -\mathbf{q}^T \mathbf{f} \\ 
&\mbox{s.t.}& \mathbf{x}^T (-\mathbf{A}) - \mathbf{q}^T \mathbf{F} \leq \mathbf{0} \\
& & \mathbf{x}^T \mathbf{E}^T = \mathbf{e}^T \\
& & \mathbf{x} \geq \mathbf{0}\\
\end{array} 
\]

\[
\begin{array}{rrl} 
&\min_{\mathbf{y},\mathbf{p}}& \mathbf{e}^T \mathbf{p} \\ 
&\mbox{s.t.}& -\mathbf{A} \mathbf{y} + \mathbf{E}^T \mathbf{p} \geq \mathbf{0} \\
& & -\mathbf{F} \mathbf{y} = -\mathbf{f} \\
& & \mathbf{y} \geq \mathbf{0}\\
\end{array} 
\]

Let $Q_{-i}$ denote the set of opponent action sequences.
For each sequence $r \in Q_{-i}$, let $y_r$ denote its
realization probability. We specify a prior density on the
feasible realization-plan polytope
$\{\mathbf{y} : \mathbf{F}\mathbf{y}=\mathbf{f},
\mathbf{y}\geq\mathbf{0}\}$, with respect to Lebesgue measure
on its affine hull, by
\[
\pi(\mathbf{y}) \propto
\prod_{r \in Q_{-i}} y_r^{\alpha_r-1},
\qquad \alpha_r \geq 1.
\]
Here $\alpha_r$ are sequence-form prior parameters, and the
MAP estimates considered below are defined using the posterior
density in realization-plan coordinates.

This prior corresponds to independent Dirichlet distributions
over the behavioral action probabilities at each opponent
information set, with appropriately chosen parameters.
Let $\rho(I)$ denote the opponent's action sequence preceding
information set $I$, and let $A(I)$ denote its available actions.
For each nonempty sequence $r$, let $\beta_r$ denote the
behavioral Dirichlet parameter for its final action.
Given the sequence-form parameters $\alpha_r$, these parameters
can be computed backward from terminal action sequences using
\[
\beta_r
=
\alpha_r
+
\sum_{I:\rho(I)=r}
\left(
\sum_{a\in A(I)} \beta_{ra}-1
\right),
\]
where $ra$ denotes the sequence obtained by appending action $a$
to $r$. This relation accounts for the Jacobian of the change
from behavioral probabilities to realization-plan coordinates.
When all $\alpha_r \geq 1$, the resulting parameters satisfy
$\beta_r \geq 1$. Conversely, arbitrary behavioral Dirichlet
parameters need not induce sequence-form parameters satisfying
$\alpha_r \geq 1$.

In our Kuhn poker experiments below, the opponent plays as player 2
and acts at most once per hand. Consequently, the sums above
are empty and $\beta_r=\alpha_r$, so the sequence-form prior
coincides with the independent behavioral Dirichlet prior
using the same parameters.

Suppose that at time $t$ we arrive at leaf node $\ell^t.$ The set of trajectories consistent with player $i$'s observations is then $o_i(\ell^t)$.
The probability of observing a compatible trajectory is proportional to the product of the realization probability assigned by the opponent strategy and the probability of the associated chance events.
The likelihood function is 
$$\frac{\sum_{r \in o_i(\ell^t)} p_r y_r}{\sum_{r \in o_i(\ell^t)} p_r},$$
where $p_r$ is the product of the chance probabilities along sequence $r.$
To simplify presentation we define $q_r$ as the normalized value:
$$q_r = \frac{p_r}{\sum_{r' \in o_i(\ell^t)} p_{r'}}.$$
Then the likelihood function is 
$$\sum_{r \in o_i(\ell^t)} q_r y_r.$$
So the posterior distribution is proportional to:

$$\left(\prod_{r \in Q_{-i}} y^{\alpha_r - 1}_r\right) \left(\prod_t \left(\sum_{r \in o_i(\ell^t)} q_r y_r\right)\right).$$
The log of the posterior is equal to the following (plus a constant):
$$\sum_r (\alpha_r -1) \log(y_r) + \sum_t \log \left(\sum_{r \in o_i(\ell^t)} q_r y_r\right)$$

The problem of finding the realization probabilities $y_r$ that maximize the posterior distribution can therefore be formulated as:

\begin{equation}\label{eq:opt}
\begin{array}{rrl} 
&\max_{\mathbf{y}} & \sum_{r \in Q_{-i}} (\alpha_r -1) \log(y_r) + \sum_t \log \left(\sum_{r \in o_i(\ell^t)} q_r y_r\right)\\ 
&\mbox{s.t.}& \mathbf{F} \mathbf{y} = \mathbf{f} \\
& & \mathbf{y} \geq \mathbf{0}\\
\end{array} 
\end{equation}

\begin{proposition}
\label{pr:concave}
If $\alpha_r \geq 1$ for all $r$, then the optimization problem (\ref{eq:opt}) is a concave maximization problem.
\end{proposition}
\begin{proof}
The general form for a concave maximization problem is the following:
\begin{equation}
\begin{array}{rrl} 
&\max_{\mathbf{x}} & f_0(x)\\ 
&\mbox{s.t.}& f_i(\mathbf{x}) \leq 0 \mbox{ for } i = 1,\ldots,m\\
& & a^T_i x = b_i \mbox{ for all } i= 1\ldots,m \\
\end{array} 
\end{equation}
where $f_0$ is concave, $f_1,\ldots,f_m$ are convex, and the equality constraints are affine.
First note that the equality constraints in our formulation are $\mathbf{F} \mathbf{y} = \mathbf{f},$
which are affine. The inequality constraints can be rewritten as $-y_r \leq 0.$ Since $-y_r$ is an affine 
function of $y_r,$ it is convex. So all inequality constraints are convex. 

Since the logarithm function is concave, $\log(y_r)$ is concave for each $r.$ Since $\alpha_r - 1 \geq 0,$
$(\alpha_r -1) \log(y_r)$ is a nonnegative multiple of a concave function and is therefore concave.
Therefore, $\sum_r (\alpha_r -1) \log(y_r)$ is concave, since it is the sum of concave functions.

For the second objective term, $\sum_{r \in o_i(\ell^t)} q_r y_r$ is affine. 
Since log is concave, by the composition rule $\log \left(\sum_{r \in o_i(\ell^t)} q_r y_r\right)$ is concave for each $t.$
So $\sum_t \log \left(\sum_{r \in o_i(\ell^t)} q_r y_r\right)$ is concave.
Therefore, the objective is the sum of two concave functions and is concave. 
\end{proof}

Note that we can equivalently characterize the problem as a convex minimization problem by negating the objective.
Such problems have the property that any local minimum is also a global minimum, and standard projected gradient 
descent methods converge to the optimal solution under appropriate step-size conditions. Since our formulation has 
constraints, we must apply projected gradient descent to ensure that the solutions maintain feasibility.

Let $f(\mathbf{y})$ denote the minimization formulation objective.
The gradient of $f$ has $r$th component:
$$\frac{1-\alpha_r}{y_r} - \sum_{t : r \in o_i(\ell^t)} \frac{q_r}{\sum_{r' \in o_i(\ell^t)} q_{r'} y_{r'}}$$

The gradient step of our algorithm to obtain the value $\mathbf{y}^{(k+1)}$ is:
$$\mathbf{z}^{(k)} = \mathbf{y}^{(k)} - \eta_k \nabla f(\mathbf{y}^{(k)})$$
The projection step is then:
$$\mathbf{y}^{(k+1)} = \argmin_{\mathbf{y} : \mathbf{F} \mathbf{y} = \mathbf{f}, \mathbf{y} \geq \mathbf{0}} \|\mathbf{y} - \mathbf{z}^{(k)} \|^2_2$$
Each projection step therefore involves solving a convex quadratic program, which can be done in polynomial time.

\begin{proposition}[Consistency under Persistent Excitation]
\label{pr:consistency}
Let $\mathcal{M}$ denote the feasible opponent realization-plan
polytope, with the prior density defined above and
$\alpha_r \geq 1$ for all $r$. Let $\mathbf{y}^* \in \mathcal{M}$
be the opponent's fixed true realization plan.

Let $S_i$ denote our finite set of pure strategies. At iteration
$t$, we select $s_t \in S_i$, possibly randomly and depending on
past observations. Conditional on this selection, the game is
played against $\mathbf{y}^*$ with fresh randomization.
For each $s \in S_i$, let $P_s(\mathbf{y})$ denote the
distribution of our observations when we play $s$ against
opponent realization plan $\mathbf{y}$.

Assume that there is a subset $S_0 \subseteq S_i$ such that:
\begin{enumerate}
\item The map
$\mathbf{y} \mapsto (P_s(\mathbf{y}))_{s \in S_0}$
is injective on $\mathcal{M}$.
\item Each $s \in S_0$ is selected with positive limiting
lower frequency:
\[
\liminf_{T\to\infty} \frac{N_s(T)}{T}>0
\quad\text{almost surely},
\qquad
N_s(T)=\sum_{t=1}^T \mathbf{1}\{s_t=s\}.
\]
\end{enumerate}
Then the posterior concentrates at $\mathbf{y}^*$, and every
sequence of global MAP estimates converges almost surely to
$\mathbf{y}^*$.
\end{proposition}

\begin{proof}
Write $p_s(o\mid\mathbf{y})$ for the probability of observation
$o$ under $P_s(\mathbf{y})$, and let $N_{s,o}(T)$ count
iterations in which we select $s$ and observe $o$.
These probabilities are affine and continuous in $\mathbf{y}$.
The strong law for bounded martingale differences gives
\[
\frac{N_{s,o}(T)}{T}
-
\frac{N_s(T)}{T}p_s(o\mid\mathbf{y}^*)
\longrightarrow 0
\quad\text{almost surely}.
\]

Fix a sample path on which these limits and the frequency
assumption hold. From any subsequence of iteration counts,
extract a further subsequence along which
$N_s(T)/T \to w_s$ for every $s$.
Then $w_s>0$ for every $s\in S_0$.
The limiting normalized log-likelihood is
\[
Q_w(\mathbf{y})
=
\sum_s w_s \sum_o p_s(o\mid\mathbf{y}^*)
                 \log p_s(o\mid\mathbf{y}),
\]
with zero-weight terms omitted. Moreover,
\[
Q_w(\mathbf{y}^*)-Q_w(\mathbf{y})
=
\sum_s w_s
D_{\mathrm{KL}}\!\left(
P_s(\mathbf{y}^*)\,\|\,P_s(\mathbf{y})
\right).
\]
Identifiability therefore makes $\mathbf{y}^*$ the unique
maximizer of $Q_w$.

Let $L_T$ denote the log-likelihood. Since there are finitely
many strategy-observation pairs, $L_T/T$ converges uniformly
to $Q_w$ on compact subsets of the relative interior of
$\mathcal{M}$. Identically impossible observations can be
omitted. At the boundary, the nonpositivity of log probabilities
and convergence of their empirical coefficients imply, for
every closed set $C\subseteq\mathcal{M}$,
\[
\limsup_{T\to\infty}
\sup_{\mathbf{y}\in C}\frac{L_T(\mathbf{y})}{T}
\leq
\sup_{\mathbf{y}\in C}Q_w(\mathbf{y}).
\]

If $C$ excludes a neighborhood of $\mathbf{y}^*$, compactness
and upper semicontinuity give a strict gap between
$\sup_C Q_w$ and $Q_w(\mathbf{y}^*)$.
An interior point sufficiently close to $\mathbf{y}^*$,
and a small compact neighborhood of that point, have
$Q_w$ values strictly above $\sup_C Q_w$.
The prior density is bounded on $\mathcal{M}$ and bounded
away from zero on this interior neighborhood.
Consequently, the posterior mass of $C$ tends to zero
exponentially along the subsequence, and no global MAP
estimate lies in $C$ for sufficiently large $T$.
This argument also covers a true realization plan on the
boundary of $\mathcal{M}$.

Every subsequence has a further subsequence satisfying these
conclusions, which establishes posterior concentration and
MAP convergence for the full sequence.
\end{proof}

These assumptions require persistent collection of observations
that jointly identify the opponent's realization plan.
Merely visiting every opponent information set is insufficient.
The result permits adaptive strategy selection, but playing
a best response to the current model does not by itself
guarantee the required information collection.

By Corollary~\ref{co:posterior-mean}, playing a best response
to the posterior mean realization plan maximizes expected payoff
against the opponent's current posterior distribution.
Our algorithm uses the posterior mode as a computationally
tractable alternative to the posterior mean. Under the stated
conditions on the sequence-form prior parameters, the mode
can be computed by convex optimization, avoiding the integration
required to compute the posterior mean. Responding to the mode
need not maximize posterior expected payoff at a finite time;
however, the algorithm achieves consistency under the assumptions
of Proposition~\ref{pr:consistency}. In contrast, the fixed-sample
approach examined in Proposition~\ref{pr:consistent} does not
guarantee consistency. The opponent modeling algorithm does not
use the zero-sum property or the payoff matrices
$\mathbf{A},\mathbf{B}$, although the payoff matrices are used
to compute a best response. Thus, the algorithm also applies
to non-zero-sum games.

\section{Experiments}
\label{se:experiments}
We evaluate our approach on the game of Kuhn poker~\cite{Kuhn50:Simplified}, with game tree shown in Figure~\ref{fi:kuhn-poker}. In Kuhn poker, there is a three-card deck consisting of a King (K), Queen (Q), and Jack (J). Initially each of two players places an ante of \$1 (so there is an initial \emph{pot} of \$2). Each player is dealt one of the cards privately (while the third card is not dealt to any player). Player 1 is allowed to bet \$1 or to \emph{check}. If P1 bets, then P2 is allowed to \emph{call} and match the bet, or to \emph{fold} and forfeit the hand. If player 2 matches the bet, then the player with the winning hand wins the full pot of \$4. If player 2 folds, then player 1 wins the pot (of \$2). If player 1 chose to check in the first round, then player 2 is allowed to bet \$1 or to check. If player 2 also checks, then the player with the highest card wins the pot of \$2. If player 2 bets when facing a check, then player 1 can call or fold. We define the observability function $o_i$ using the standard rules of poker: if a player folds, then neither player observes the other player's card. If the hand finishes without any player folding (i.e., the hand goes to \emph{showdown}), then both players observe the opponent's card.\footnote{Note that in live poker if a hand goes to showdown the player ``out of position'' is often supposed to show their hand first, at which point the other player has the opportunity to ``muck'' their losing hand without having to show it. Also the first player may choose to muck their hand rather than showing it. In online poker typically both hands are always observed when neither player folds. We will follow the online convention, since the live poker rules are not always consistent on this issue.} Both players have 12 action sequences (plus the empty sequence) and 6 information sets. The sequence-form matrices $\mathbf{E}$ and $\mathbf{F}$ therefore have dimension $7 \times 13$, and are defined in Tables~\ref{ta:E} and~\ref{ta:F} (with zero entries omitted). The vectors $\mathbf{e}$ and $\mathbf{f}$ are both equal to $(1,0,0,0,0,0,0)^T.$ Matrix $\mathbf{A}$ has dimension $13 \times 13$ and is defined in Table~\ref{ta:A} (with zeroes omitted). Since the game is zero sum, the payoff matrix for player 2 is $\mathbf{B} = -\mathbf{A}.$ The observability functions $o_1$ and $o_2$ are given in Table~\ref{ta:obs}. The full game tree is given in Figure~\ref{fi:kuhn-poker}.

\begin{figure*}[!ht]
\centering
\includegraphics[width=\textwidth,height=\textheight,keepaspectratio]{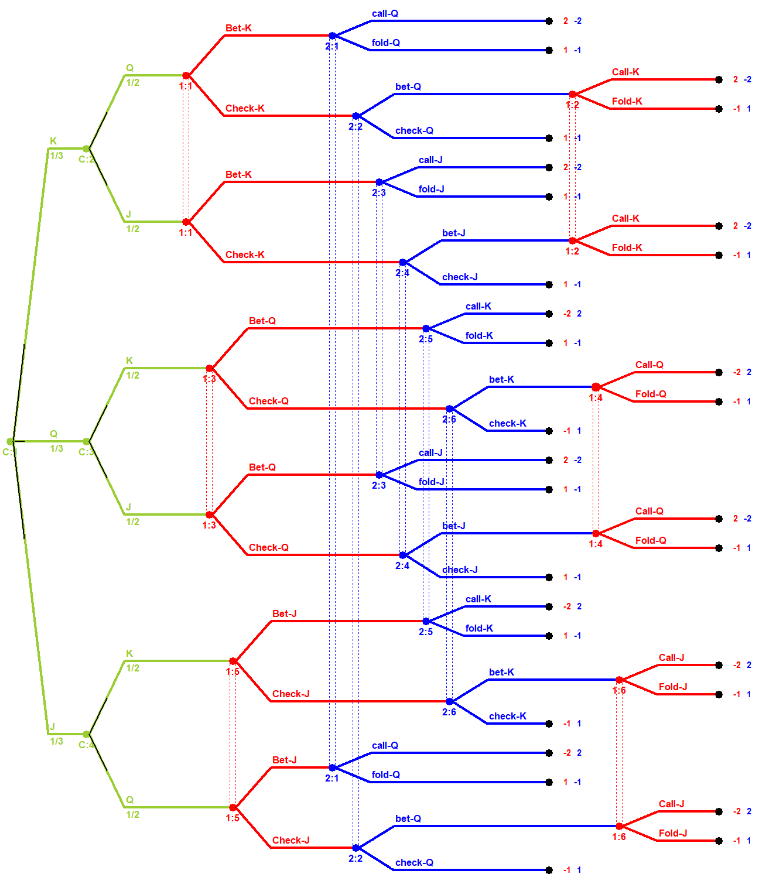}
\caption{Game tree for Kuhn poker.}
\label{fi:kuhn-poker}
\end{figure*}

\renewcommand{\tabcolsep}{4pt}
\begin{table*}[!ht]
\centering
\caption{Matrix $\mathbf{E}$ for Kuhn poker.}
\label{ta:E}
\begin{tabular}{|*{13}{c}|}
$\emptyset$ &$B_K$ &$Ch_K$ &$Ch_KCa_K$ &$Ch_KF_K$ &$B_Q$ &$Ch_Q$ &$Ch_QCa_Q$ &$Ch_QF_Q$ &$B_J$ &$Ch_J$ &$Ch_JCa_J$ &$Ch_JF_J$ \\ \hline
1 & & & & & & & & & & & & \\
-1 &1 &1 & & & & & & & & & & \\
& &-1 &1 &1 & & & & & & & & \\
-1 & & & & &1 &1 & & & & & & \\
& & & & & &-1 &1 &1 & & & & \\
-1 & & & & & & & & &1 &1 & & \\
& & & & & & & & & &-1 &1 &1 \\ \hline
\end{tabular}
\end{table*}

\renewcommand{\tabcolsep}{6pt}
\begin{table*}[!ht]
\centering
\caption{Matrix $\mathbf{F}$ for Kuhn poker.}
\label{ta:F}
\begin{tabular}{|*{13}{c}|}
$\emptyset$ &$ca_Q$ &$f_Q$ &$b_Q$ &$ch_Q$ &$ca_J$ &$f_J$ &$b_J$ &$ch_J$ &$ca_K$ &$f_K$ &$b_K$ &$ch_K$ \\ \hline
1 & & & & & & & & & & & & \\
-1 &1 &1 & & & & & & & & & & \\
-1& & &1 &1 & & & & & & & & \\
-1& & & & &1 &1 & & & & & & \\
-1& & & & & & &1 &1 & & & & \\
-1& & & & & & & & &1 &1 & & \\
-1& & & & & & & & & & &1 &1 \\\hline
\end{tabular}
\end{table*}

\begin{table*}[!ht]
\centering
\caption{Matrix $\mathbf{A}$ for Kuhn poker.}
\label{ta:A}
\begin{tabular}{c|*{13}{c}|}
 &$\emptyset$ &$ca_Q$ &$f_Q$ &$b_Q$ &$ch_Q$ &$ca_J$ &$f_J$ &$b_J$ &$ch_J$ &$ca_K$ &$f_K$ &$b_K$ &$ch_K$ \\ \hline
$\emptyset$ & & & & & & & & & & & & & \\ 
$B_K$ & &$\frac{1}{3}$ &$\frac{1}{6}$ & & &$\frac{1}{3}$ &$\frac{1}{6}$ & & & & & & \\ 
$Ch_K$ & & & & &$\frac{1}{6}$ & & & &$\frac{1}{6}$ & & & & \\ 
$Ch_KCa_K$ & & & &$\frac{1}{3}$ & & & &$\frac{1}{3}$ & & & & & \\ 
$Ch_KF_K$ & & & &$-\frac{1}{6}$ & & & &$-\frac{1}{6}$ & & & & & \\ 
$B_Q$ & & & & & &$\frac{1}{3}$ &$\frac{1}{6}$ & & &$-\frac{1}{3}$ &$\frac{1}{6}$ & & \\ 
$Ch_Q$ & & & & & & & & &$\frac{1}{6}$ & & & &$-\frac{1}{6}$ \\ 
$Ch_QCa_Q$ & & & & & & & &$\frac{1}{3}$ & & & &$-\frac{1}{3}$ & \\ 
$Ch_QF_Q$ & & & & & & & &$-\frac{1}{6}$ & & & &$-\frac{1}{6}$ & \\ 
$B_J$ & &$-\frac{1}{3}$ &$\frac{1}{6}$ & & & & & & &$-\frac{1}{3}$ &$\frac{1}{6}$ & & \\ 
$Ch_J$ & & & & &$-\frac{1}{6}$ & & & & & & & &$-\frac{1}{6}$ \\ 
$Ch_JCa_J$ & & & &$-\frac{1}{3}$ & & & & & & & &$-\frac{1}{3}$ & \\ 
$Ch_JF_J$ & & & &$-\frac{1}{6}$ & & & & & & & &$-\frac{1}{6}$ & \\ \hline
\end{tabular}
\end{table*}

\begin{table*}[!ht]
\centering
\caption{Observability functions $o_1$ and $o_2$.}
\label{ta:obs}
\begin{tabular}{c|c|c}
Leaf node $\ell$ &$o_1(\ell)$ &$o_2(\ell)$ \\ \hline \hline
$B_Kca_Q$ & $\{B_Kca_Q\}$ &$\{B_Kca_Q\}$\\ \hline
$B_Kf_Q$ & $\{B_Kf_Q, B_Kf_J\}$ &$\{B_Kf_Q, B_Jf_Q\}$\\ \hline
$Ch_Kb_QCa_K$ &$\{Ch_Kb_QCa_K\}$ &$\{Ch_Kb_QCa_K\}$ \\ \hline
$Ch_Kb_QF_K$ &$\{Ch_Kb_QF_K, Ch_Kb_JF_K\}$ &$\{Ch_Kb_QF_K, Ch_Jb_QF_J\}$ \\ \hline
$Ch_Kch_Q$ &$\{Ch_Kch_Q\}$ &$\{Ch_Kch_Q\}$ \\ \hline
$B_Kca_J$ &$\{B_Kca_J\}$ &$\{B_Kca_J\}$ \\ \hline
$B_Kf_J$ &$\{B_Kf_J, B_Kf_Q\}$ &$\{B_Kf_J, B_Qf_J\}$ \\ \hline
$Ch_Kb_JCa_K$ &$\{Ch_Kb_JCa_K\}$ &$\{Ch_Kb_JCa_K\}$ \\ \hline
$Ch_Kb_JF_K$ &$\{Ch_Kb_JF_K, Ch_Kb_QF_K\}$ &$\{Ch_Kb_JF_K, Ch_Qb_JF_Q\}$ \\ \hline
$Ch_Kch_J$ &$\{Ch_Kch_J\}$ &$\{Ch_Kch_J\}$ \\ \hline
$B_Qca_K$ & $\{B_Qca_K\}$ &$\{B_Qca_K\}$\\ \hline
$B_Qf_K$ & $\{B_Qf_K, B_Qf_J\}$ &$\{B_Qf_K, B_Jf_K\}$\\ \hline
$Ch_Qb_KCa_Q$ &$\{Ch_Qb_KCa_Q\}$ &$\{Ch_Qb_KCa_Q\}$ \\ \hline
$Ch_Qb_KF_Q$ &$\{Ch_Qb_KF_Q, Ch_Qb_JF_Q\}$ &$\{Ch_Qb_KF_Q, Ch_Jb_KF_J\}$ \\ \hline
$Ch_Qch_K$ &$\{Ch_Qch_K\}$ &$\{Ch_Qch_K\}$ \\ \hline
$B_Qca_J$ &$\{B_Qca_J\}$ &$\{B_Qca_J\}$ \\ \hline
$B_Qf_J$ &$\{B_Qf_J, B_Qf_K\}$ &$\{B_Qf_J, B_Kf_J\}$ \\ \hline
$Ch_Qb_JCa_Q$ &$\{Ch_Qb_JCa_Q\}$ &$\{Ch_Qb_JCa_Q\}$ \\ \hline
$Ch_Qb_JF_Q$ &$\{Ch_Qb_JF_Q, Ch_Qb_KF_Q\}$ &$\{Ch_Qb_JF_Q, Ch_Kb_JF_K\}$ \\ \hline
$Ch_Qch_J$ &$\{Ch_Qch_J\}$ &$\{Ch_Qch_J\}$ \\ \hline
$B_Jca_K$ & $\{B_Jca_K\}$ &$\{B_Jca_K\}$\\ \hline
$B_Jf_K$ & $\{B_Jf_K, B_Jf_Q\}$ &$\{B_Jf_K, B_Qf_K\}$\\ \hline
$Ch_Jb_KCa_J$ &$\{Ch_Jb_KCa_J\}$ &$\{Ch_Jb_KCa_J\}$ \\ \hline
$Ch_Jb_KF_J$ &$\{Ch_Jb_KF_J, Ch_Jb_QF_J\}$ &$\{Ch_Jb_KF_J, Ch_Qb_KF_Q\}$ \\ \hline
$Ch_Jch_K$ &$\{Ch_Jch_K\}$ &$\{Ch_Jch_K\}$ \\ \hline
$B_Jca_Q$ &$\{B_Jca_Q\}$ &$\{B_Jca_Q\}$ \\ \hline
$B_Jf_Q$ &$\{B_Jf_Q, B_Jf_K\}$ &$\{B_Jf_Q, B_Kf_Q\}$ \\ \hline
$Ch_Jb_QCa_J$ &$\{Ch_Jb_QCa_J\}$ &$\{Ch_Jb_QCa_J\}$ \\ \hline
$Ch_Jb_QF_J$ &$\{Ch_Jb_QF_J, Ch_Jb_KF_J\}$ &$\{Ch_Jb_QF_J, Ch_Kb_QF_K\}$ \\ \hline
$Ch_Jch_Q$ &$\{Ch_Jch_Q\}$ &$\{Ch_Jch_Q\}$ \\ \hline
\end{tabular}
\end{table*}

In our experiments we will assume a Dirichlet prior distribution with $\alpha_i$ initially 
equal to 2 for each action at each information set. Thus, the mean (as well as the mode) 
of the prior distribution selects each action with equal probability. Using initial parameter
values of 2 is a common choice in prior work, e.g.,~\cite{Southey05:Bayes}. Note that these values
satisfy the conditions from Proposition~\ref{pr:concave} that $\alpha_i \geq 1.$
For each set of experiments we will sample the opponent's strategy $\boldsymbol{\sigma}^*_{-i}$ from the prior distribution.
We will compare our new algorithm against several benchmark approaches. We denote our algorithm as FMAP for ``full
max a posteriori.'' We will compare against three previously considered sampling approaches:
Bayesian Best Response (BBR), Max a Posteriori Response (MAP), and Thompson's Response~\cite{Southey05:Bayes}. 
Note that none of these sampling approaches are consistent. In the experiments we will play the role of player 1.
We also compare performance against the Nash equilibrium strategy that performs best against $\boldsymbol{\sigma}^*_{-i}$ (player
1 has infinitely many Nash equilibrium strategies~\cite{Kuhn50:Simplified}). 
Note that the best Nash equilibrium has an ``unfair advantage'' in that we would
not actually know which Nash equilibrium performs best against $\boldsymbol{\sigma}^*_{-i}$ when we are playing. 
As a final benchmark we will compare against the actual best response to $\boldsymbol{\sigma}^*_{-i}$, which again has an
``unfair advantage'' since it would not be known during gameplay. We will compare performance of all approaches
against several different strategies for the opponent sampled from the prior distribution.  
For each of the opponent modeling algorithms, we can compute a best response to the modeled strategy at each iteration
either using expectimax or by solving the following linear program, where $\mathbf{y}$ denotes
player 2's modeled strategy in sequence form~\cite{Koller94:Fast}:

\[
\begin{array}{rrl} 
&\max_{\mathbf{x}}& \mathbf{x}^T (\mathbf{A} \mathbf{y}) \\ 
&\mbox{s.t.}& \mathbf{x}^T \mathbf{E}^T = \mathbf{e}^T \\
& & \mathbf{x} \geq \mathbf{0}\\
\end{array} 
\]

In our experiments we generate 100 opponents randomly from the prior distribution. We run each of the approaches against
each opponent for $T = 3,000$ iterations. Each of the sampling algorithms (BBR, MAP, and Thompson) samples $k = 10$ strategies
at the start of the match against each opponent. We use several different techniques to reduce variance in our experiments. 
First, the sampled strategies are the same for all three of the algorithms (though a separate set of $k$ sampled strategies 
is used against each opponent). We assume the cards dealt to the players are the same for all of the opponent modeling algorithms
at each iteration. At each iteration, we assume that the strategies played by all opponent modeling algorithms, as well as the true strategy
played by player 2, use the same randomization thresholds for their strategies at all information sets. For example, at iteration 100 of the experiments, suppose we deal player 1
King and player 2 Jack for all of the algorithms, and we ``deal'' player 1 a value of 0.68 for the randomization threshold for his initial information set with a King.
Then the algorithms for which player 1 bets with a King with probability above 0.68 will bet while the other algorithms will check. 
Additionally, at each iteration we calculate the expected payoff between each pair of strategies as opposed to the actual payoff. Since we know the full strategy
for player 1 at time $t$, call it $\mathbf{x},$ as well as the true strategy for player 2, call it $\mathbf{y},$ we can compute the true expected payoff by $\mathbf{x}^T \mathbf{A} \mathbf{y},$
rather than using just a single sampled value. However, we still play out the hand so that the opponent modeling algorithms are able to update their model based on observations of the 
opponent's play. For the fixed strategies BestResponse and BestNash, we simply calculate their expected payoff against the target strategy at the outset and do not need to simulate gameplay.

The results are given in Figure~\ref{fi:results}. The $x$-axis is the game iteration $t,$ and the $y$-axis is the expected profit achieved at time $t$ averaged over all of the 
opponent strategies. The BestResponse strategy has a constant value of 0.576, and BestNash has a constant value of 0.173. The game value to player 1 is $-\frac{1}{18} = -0.056$~\cite{Kuhn50:Simplified}, so BestNash is still beating the opponents for significantly more than the game value. All of the opponent modeling algorithms perform significantly better than BestNash starting at the first game iteration. All three of the sampling algorithms perform very similarly: the average expected payoff of BBR at the final iteration is 0.557, of Thompson is 0.547, and of MAP is 0.537. (It has been observed previously that these algorithms perform similarly when playing against opponents drawn from the prior~\cite{Southey05:Bayes}). We can see that after the first few hundred hands there is little change in their average per-period performance.
Our new algorithm FMAP quickly surpasses the performance of the sampling algorithms within the first hundred hands, ultimately achieving average expected payoff of 0.573 at the final iteration. FMAP clearly outperforms the sampling approaches, and nearly converges to the BestResponse value.

\begin{figure*}[!ht]
\centering
\includegraphics[width=\textwidth,height=\textheight,keepaspectratio]{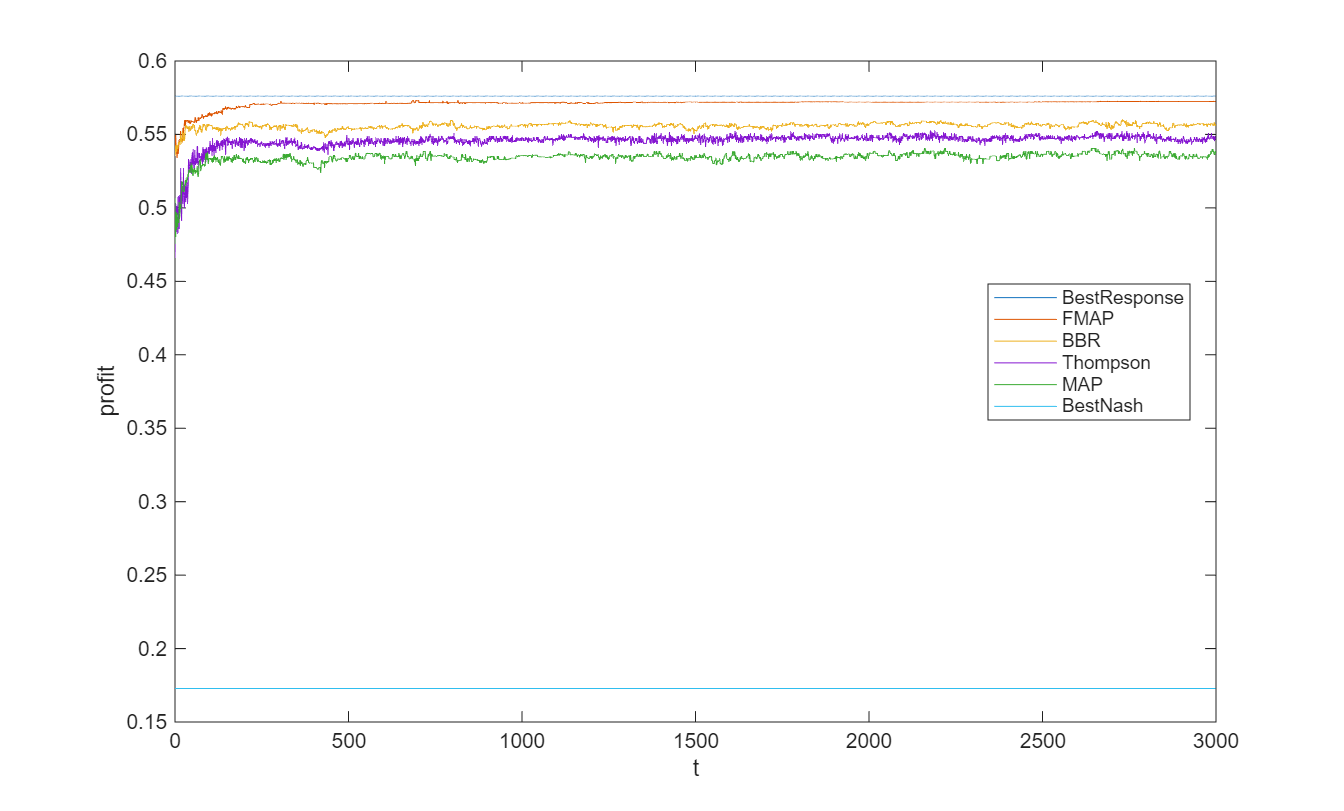}
\caption{Expected payoff as a function of game iteration for several opponent modeling algorithms and benchmark strategies. The results are averaged over 100 opponents generated randomly from the prior distribution. The sampling algorithms all use 10 samples.}
\label{fi:results}
\end{figure*}

We used relatively standard values for several hyperparameters in the implementation of gradient descent for the FMAP algorithm. We ran gradient descent until the norm of the difference in iterate solutions was less than $\num{1e-7}$ or for a maximum of 1000 iterations (typically far fewer than 1000 iterations were used). We initialized the learning rate $\eta$ to be 1, which was modified according to the Armijo condition in backtracking line search using hyperparameters $c = \num{1e-4},$ $\beta = 0.5$~\cite{Armijo66:Minimization} (we repeatedly multiplied $\eta$ by $\beta$ until the required condition was satisfied). We halted the algorithm if $\eta$ ever went below a ``safeguard'' value of $\num{1e-16}.$ We solved the projection subproblem as a convex quadratic program using Gurobi version 12.0.3~\cite{Gurobi25:Gurobi} on a laptop using Windows 11. In order to ensure that the subproblem did not return any vectors with component equal to 0, we used constraints $y_i \geq \num{1e-6}$ instead of $y_i \geq 0.$ 

\section{Related research}
\label{se:related}
While we have used a full best response strategy for all our algorithms, note that the meta-algorithm presented as Algorithm~\ref{al:meta} allows for arbitrary response function $r_i.$ A best response may
not be robust if the opponent's true strategy differs significantly from the opponent model, and other approaches may be preferable. Potentially more robust alternatives include a restricted Nash response~\cite{Johanson07:Computing} and related concept of an $\epsilon$-safe best response~\cite{McCracken04:Safe}. These approaches balance between exploitation and safety, maximizing performance against the opponent's strategy subject to a constraint on our strategy's exploitability. Note that these approaches are specific to two-player zero-sum games, while our new algorithm also applies to non-zero-sum games. An alternative approach for achieving robustness may be to apply counterfactual regret minimization (CFR)~\cite{Zinkevich07:Regret}. CFR maintains counterfactual regret values for each action at each information set, and a popular variant uses regret matching to select the actions in proportion to their positive counterfactual regret at each iteration; this may be more robust than playing a pure best response. Note that CFR is not an alternative to our opponent modeling algorithm; it would just provide an alternative procedure to responding to the opponent model. The application of standard CFR in real time would require the aid of an opponent modeling algorithm to produce baseline counterfactual values for unobserved paths of play. For example, if we are player 1 in Kuhn poker and bet with a King and player 2 folds, that observation alone does not reveal how player 2 would respond had we checked instead; an opponent model could be used to estimate the corresponding counterfactual value.

\section{Conclusion}
\label{se:conclusion}
In this paper we studied perhaps the simplest possible setting for opponent modeling in imperfect-information games: two-player games where the opponent is playing a static fixed strategy that is drawn from a known prior distribution. We showed that even in this simple setting, prior approaches fail to satisfy the very natural property that the opponent model converges to the opponent's true strategy in the limit as the number of game iterations approaches infinity. We refer to this property as \emph{consistency}, which is closely related to the concept of consistency from the field of statistics. We showed that the previously most successful algorithm, Bayesian Best Response, is not consistent, even if the target strategy is in the convex hull of the sampled strategies. We also introduced the concept of an \emph{observability function} which is needed to define and reason about repeated imperfect-information games. Several prior opponent modeling approaches assume that opponents' private information or actions off the path of play are always observed between game iterations, ignoring important nuances that occur in many real-world situations where information is partially observed.

We introduced a new algorithm that computes a maximum a
posteriori opponent strategy in realization-plan coordinates
under the sequence-form prior defined above. This prior
corresponds to independent Dirichlet distributions at the
opponent's information sets, with parameters chosen so that
the induced sequence-form parameters satisfy $\alpha_r \geq 1$.
The Dirichlet prior is very common for modeling probability vectors in a multidimensional simplex. With no further information available, we could choose a Dirichlet distribution with equal parameters for each action (whose mean selects each action with equal probability). If historical information is available, we could alternatively use a Dirichlet prior distribution whose mean strategy is calculated from the data. An additional option that has proven successful in certain applications is to use a Dirichlet prior whose mean strategy is a Nash equilibrium (or approximation of one)~\cite{Ganzfried11:Game,Ganzfried24:Opponent}. Our algorithm runs efficiently by solving a convex minimization problem using projected gradient descent. We show that our algorithm achieves consistency under the stated
identifiability and persistent-excitation assumptions, unlike several prior algorithms that utilize sampling.

We compared our algorithm against several other previously successful algorithms on Kuhn poker, a well-studied imperfect-information game. We demonstrated that our algorithm outperformed three successful sampling-based algorithms and nearly converged to the full best response. All of the opponent modeling algorithms significantly outperformed using the Best Nash equilibrium against the target strategy, illustrating the potential improvements that can be obtained by a successful opponent modeling algorithm over following classic game-theoretic solution concepts.

While we experimented on a two-player zero-sum game, our algorithm applies to non-zero-sum games as well. Since our algorithm is based on solving a convex minimization problem using a standard version of projected gradient descent, we expect the approach to be scalable to large problems of size comparable to other convex minimization problems solved by gradient descent. While the sampling algorithms run very quickly when a small number of samples is used, the number of samples required to obtain strong performance may grow exponentially in the size of the game tree. Adaptively adding in new sampled strategies is unlikely to improve performance, since the computational cost of computing posterior values for newly-added sampled strategies would be the same as if we had maintained them throughout all game iterations. Our algorithm could be extended to more than two players, though the optimization would no longer be a convex minimization problem due to products of variables. Now that we have a better understanding and successful approach for the setting of static opponents from a known prior distribution, we would like to investigate settings where the opponent plays a strategy that is not drawn from the prior distribution or is dynamic.

\bibliographystyle{plain}
\bibliography{C://FromBackup/Research/refs/dairefs}

\end{document}